\documentclass[sigconf,authorversion,nonacm, screen=true]{acmart}

\usepackage[utf8]{inputenc}

\title[Node-DP Estimation of the Number of Connected Components]{Node-Differentially Private Estimation of the Number of Connected Components}

\author{Iden Kalemaj}
\affiliation{Boston University \country{US}}

\email{ikalemaj@bu.edu}

\author{Sofya Raskhodnikova}
\affiliation{Boston University \country{US}}
\email{sofya@bu.edu}

\author{Adam Smith}
\affiliation{Boston University \country{US} }
\email{ads22@bu.edu}

\author{Charalampos E. Tsourakakis}
\affiliation{Boston University \country{US}}
\email{tsourolampis@gmail.com}

\newcommand{\notes}{0}

\usepackage{amsmath,amsthm,amsfonts, verbatim,relsize,mathrsfs,hyperref,soul, url}
\usepackage{algorithm}
\usepackage[noend]{algpseudocode}
\hypersetup{colorlinks = true,linkcolor = violet, anchorcolor = violet, citecolor = violet, filecolor = red,urlcolor = red}
\usepackage [english]{babel}
\usepackage [autostyle, english = american]{csquotes}
\usepackage{blindtext}
\usepackage{framed}
\usepackage{fullpage}
\usepackage[shortlabels]{enumitem}
\usepackage[normalem]{ulem}
\usepackage{graphicx}
\usepackage{xcolor}
\usepackage{subcaption}
\usepackage{subfiles}
\usepackage{dsfont}
\usepackage{bbm}

\usepackage{setspace}

\MakeOuterQuote{"}

\makeatletter
\newcommand\footnoteref[1]{\protected@xdef\@thefnmark{\ref{#1}}\@footnotemark}
\makeatother

\newcommand*\Let[2]{\State #1 $\gets$ #2}
\algrenewcommand\algorithmicrequire{\textbf{Input:}}
\algrenewcommand\algorithmicensure{\textbf{Output:}}

\newtheorem{theorem}{Theorem}[section]

\newtheorem{definition}[theorem]{Definition}
\newtheorem{lemma}[theorem]{Lemma}

\newtheorem{remark}[theorem]{Remark}

\newtheorem{claim}[theorem]{Claim}

\newcommand{\Sec}[1]{\hyperref[sec:#1]{Section~\ref*{sec:#1}}} 
\newcommand{\Appendix}[1]{\hyperref[sec:#1]{Appendix~\ref*{sec:#1}}} 

\newcommand{\Eqn}[1]{\hyperref[eq:#1]{(\ref*{eq:#1})}} 
\newcommand{\Fig}[1]{\hyperref[fig:#1]{Fig.\,\ref*{fig:#1}}} 
\newcommand{\Tab}[1]{\hyperref[tab:#1]{Tab.\,\ref*{tab:#1}}} 
\newcommand{\Thm}[1]{\hyperref[thm:#1]{Theorem\,\ref*{thm:#1}}} 
\newcommand{\Fact}[1]{\hyperref[fact:#1]{Fact\,\ref*{fact:#1}}} 
\newcommand{\Lem}[1]{\hyperref[lem:#1]{Lemma\,\ref*{lem:#1}}} 
\newcommand{\Lems}[2]{\hyperref[lem:#1]{Lemmas\,\ref*{lem:#1}} and~\hyperref[lem:#2]{\ref*{lem:#2}}} 
\newcommand{\Prop}[1]{\hyperref[prop:#1]{Prop.\,\ref*{prop:#1}}} 
\newcommand{\Cor}[1]{\hyperref[cor:#1]{Corollary~\ref*{cor:#1}}} 
\newcommand{\Conj}[1]{\hyperref[conj:#1]{Conjecture~\ref*{conj:#1}}} 
\newcommand{\Def}[1]{\hyperref[def:#1]{Definition~\ref*{def:#1}}} 
\newcommand{\Alg}[1]{\hyperref[alg:#1]{Algorithm~\ref*{alg:#1}}} 
\newcommand{\Proc}[1]{\hyperref[proc:#1]{Procedure~\ref*{proc:#1}}} 
\newcommand{\Step}[1]{\hyperref[step:#1]{Step~\ref*{step:#1}}} 
\newcommand{\Steps}[2]{\hyperref[step:#1]{Steps~\ref*{step:#1}} and~\hyperref[step:#2]{\ref*{step:#2}}} 
\newcommand{\Stepss}[3]{\hyperref[step:#1]{Steps~\ref*{step:#1}},~\hyperref[step:#2]{\ref*{step:#2}}, and~\hyperref[step:#3]{\ref*{step:#3}}} 
\newcommand{\Ex}[1]{\hyperref[ex:#1]{Ex.~\ref*{ex:#1}}} 
\newcommand{\Clm}[1]{\hyperref[clm:#1]{Claim~\ref*{clm:#1}}} 
\newcommand{\Inv}[1]{\hyperref[inv:#1]{Invariant~\ref*{inv:#1}}} 
\newcommand{\Rem}[1]{\hyperref[rem:#1]{Remark~\ref*{rem:#1}}} 
\newcommand{\Obs}[1]{\hyperref[obs:#1]{Observation~\ref*{obs:#1}}} 

\newcommand{\R}{\mathbb R}

\newcommand{\eps}{\varepsilon}

\newcommand{\cA}{\mathcal{A}}
\newcommand{\cP}{\mathcal{P}}

\newcommand{\cG}{\mathcal{G}}
\newcommand{\cF}{\mathcal{F}}
\newcommand{\fsf}{f_{\mathrm{sf}}}
\newcommand{\fcc}{f_{\mathrm{cc}}}
\newcommand{\Err}{\mathrm{Err}}

\newcommand{\E}{\mathbb{E}}

\newcommand{\Bparen}[1]{\Big( {#1} \Big)}

\newcommand{\inducedstar}{s}

\ifnum \notes=1
    \newcommand{\srnote}[1]{{\color{violet}\footnote{{\color{blue} {\bf S:} #1}}}}
    \newcommand{\asnote}[1]{{\color{violet}\footnote{{\color{violet} {\bf A:} #1}}}}
    \newcommand{\inote}[1]{{\color{brown}\footnote{{\color{brown} {\bf I:} #1}}}}
    \newcommand{\btnote}[1]{{\color{brown}\footnote{{\color{red} {\bf B:} #1}}}}
    \newcommand{\sr}[1]{{\color{blue} #1}}
    \newcommand{\as}[1]{{\color{violet} #1}}

\else 
    \newcommand{\srnote}[1]{}
    \newcommand{\asnote}[1]{}
    \newcommand{\inote}[1]{}
    \newcommand{\btnote}[1]{}
    \newcommand{\sr}[1]{{#1}}
    \newcommand{\as}[1]{{#1}}

\fi

\begin{document}
\pagenumbering{arabic}
\copyrightyear{2023}
\acmYear{2023}
\setcopyright{acmlicensed}\acmConference[PODS '23]{Proceedings of the 42nd ACM SIGMOD-SIGACT-SIGAI Symposium on Principles of Database Systems}{June 18--23, 2023}{Seattle, WA, USA}
\acmBooktitle{Proceedings of the 42nd ACM SIGMOD-SIGACT-SIGAI Symposium on Principles of Database Systems (PODS '23), June 18--23, 2023, Seattle, WA, USA}
\acmPrice{15.00}
\acmDOI{10.1145/3584372.3588671}
\acmISBN{979-8-4007-0127-6/23/06}

\begin{abstract}
	We design the first node-differentially private algorithm for approximating the number of connected components in a graph. Given a database representing an $n$-vertex graph $G$ and a privacy parameter $\eps$, our algorithm runs in polynomial time and, with probability $1-o(1)$, has additive error $\widetilde{O}(\frac{\Delta^*\ln\ln n}{\eps}),$ where $\Delta^*$ is the smallest possible maximum degree of a spanning forest of $G.$
	Node-differentially private algorithms are known only for a small number of database analysis tasks. A major
	obstacle for designing such an algorithm for the number of connected components is that this graph statistic
	is not robust to adding one node with arbitrary connections (a change that node-differential privacy is designed to hide): {\em every} graph is a neighbor of a connected graph. 
	
	We overcome this by designing a family of efficiently computable Lipschitz extensions of the number of connected components or, equivalently, the size of a spanning forest. The  construction of the extensions, which is at the core of our algorithm, is based on the forest polytope of $G.$ We prove several combinatorial facts about spanning forests, in particular, that 
	a graph with no induced $\Delta$-stars
	has a spanning forest of degree at most $\Delta$. With this fact, we show that our Lipschitz extensions for the number of connected components equal the true value of the function for the largest possible monotone families of graphs. More generally, on all monotone sets of graphs, the $\ell_\infty$ error of our Lipschitz extensions is nearly optimal. 	
\end{abstract}




\settopmatter{printacmref=false, printfolios=false}

\maketitle

\renewcommand{\shortauthors}{Kalemaj, Raskhodnikova, Smith, Tsourakakis}
	
\section{Introduction}	

Counting the number of features in a graph, ranging from local structures  such as edges, triangles, wedges and other small motifs~\cite{milo2002network,xiang2022general,motifscope} to more global features such as the number of connected components~\cite{frank1978estimation,kang2010patterns,kiveris2014connected,klusowski2020estimating}, is a  family of foundational  tasks in graph mining.   Such  counts are used both in feature-based approaches to anomaly detection in graphs~\cite{akoglu2015graph} and as quality criteria for  generative models~\cite{xiang2022general}.  Specifically, the number of connected components is a basic  statistic for understanding the structure of a network at the most fundamental level.
Its uses range from determining
the number of classes in a population~\cite{goodman1949estimation}
to approximating such diverse statistics as the weight of the minimum spanning tree~\cite{ChazelleRT05} and
the number of documented deaths in the Syrian war
~\cite{chen2018unique}. 

When the underlying network data captures private relationships between people,
publishing the number of connected components in the network can leak those individuals'  sensitive information. 
\emph{Differential privacy}~\cite{DworkMNS16} is a widely studied and deployed 
formal privacy guarantee for data analysis.
The output distributions of a differentially private algorithm must look nearly indistinguishable for any two input databases that differ only in the data of a single individual. In this work, we focus on relational databases that 
store graphs where nodes represent individuals and edges capture relationships between them.
There are two natural adaptations of differential privacy for graph databases: {\em edge differential privacy} and {\em node differential privacy} (or, more concisely, {\em edge-privacy} and {\em node-privacy})~\cite{HayLMJ09}. 
For edge-privacy, first investigated by Nissim et al.~\cite{NissimRS07}, the indistinguishability requirement applies to any two graphs that differ in one edge. In contrast, for node-privacy (first studied by three concurrent works~\cite{BlockiBDS13,KasiviswanathanNRS13,ChenZ13}), it applies to any two graphs that differ in one node and all its adjacent edges. Node-privacy is more suitable for databases representing networks of people, since, in this context, it protects each individual (node) and all their connections (adjacent edges). However, node-privacy is much harder to attain than edge-privacy, since it hides much larger changes to the database and provides much stronger privacy protections.

We design the first node-private algorithm for approximating the number of connected components in a graph database. We show that our algorithm runs in polynomial time and give an upper bound on its additive error. The bound  is low for  real-world networks and typical graphs arising from random network models. 

Connectivity of graph databases has not been previously investigated in the context of node-privacy. Let $\fcc(G)$ denote the number of connected components in a graph $G$. Obtaining an edge-private algorithm for approximating $\fcc$ is trivial. Without privacy guarantees, approximating $\fcc$ with additive error has been studied in the literature on sublinear-time algorithms~\cite{ChazelleRT05,BerenbrinkKM14,klusowski2020estimating}.

\subsubsection*{Instance-based accuracy} A major obstacle for designing a node-private algorithm for $\fcc$ is that this graph statistic is sensitive to adding one node with arbitrary connections: {\em every} graph $G$ can be changed to a connected graph $G'$ by adding a new node adjacent to all nodes in $G$. By definition, the output distributions of a node-private algorithm $\cA$ must be nearly indistinguishable on $G$ and $G'$. If $\fcc(G)$ is large, $\cA$ must be inaccurate on at least one of $G$ and $G'.$
Thus, meaningful worst case guarantees for approximating $\fcc$ are impossible.

To overcome this lower bound, we provide instance-based accuracy guarantees. Our bound on the error is roughly $\Delta^*$, the smallest possible maximum degree of a spanning forest of the input database $G$. 
In particular, 
$\Delta^*$ is always at most 
the maximum degree of the graph, though it can be much smaller. 
Previous works on node-private algorithms either do not analyze accuracy at all or analyze it 
through lenses that are too coarse (for example, by looking at the maximum degree) or too task-specific for our purposes; see the discussion in \Sec{related_work}.

The core of our algorithm is a construction of a family of efficiently computable \emph{Lipschitz extensions} of $\fcc$ that provide good approximations to the true function values.

\subsubsection*{Background on Lipschitz extensions}\label{sec:prelims-on-Lipschitz-extensions}
The framework of Lipschitz extensions for designing node-private algorithms was proposed by three concurrent works \cite{BlockiBDS13,KasiviswanathanNRS13,ChenZ13} and subsequently refined~\cite{RaskhodnikovaS16,RaskhodnikovaS16-E}. 

\begin{definition}[Node-neighboring graphs; node-distance; $\Delta$-Lipschitz functions]
Two graphs $G$ and $G'$ are {\em node-neighbors} if one can be obtained from the other by removing a vertex and all of its adjacent edges. 

The {\em node-distance} $d(G, G')$ between two graphs $G$ and $G'$ is the smallest number of modifications needed to obtain $G'$ from $G$, where each modification is either the removal of one node and all of its adjacent edges, or the insertion of one node with arbitrary edges incident on it. 

Let $\cG$ denote the set of all graphs. A function $h \colon \cG \to \R$  is {\em $\Delta$-Lipschitz} on a set $S$ of graphs
if $|h(G) - h(G')| \leq \Delta \cdot d(G, G')$ for all $G,G'\in S$.
When $S=\cG,$ we refer to $f$ as $\Delta$-Lipschitz,  without specifying $S$.
\end{definition}

A $\Delta$-Lipschitz function can be approximated \as{node-privately} with error $\Theta(\Delta)$ (with high probability and for constant privacy parameter) using the standard Laplace mechanism (see \Thm{laplace}).

Given a function $f$ and a parameter $\Delta$, a Lipschitz extension $f_{\Delta}$ is defined by first identifying an {\em anchor set} of graphs on which $f$ is $\Delta$-Lipschitz.
Then $f_\Delta$ is made identical to $f$ on the anchor set and is extended outside the set so that it is $\Delta$-Lipschitz on $\cG$. Instead of approximating the sensitive function $f$, a node-private algorithm $\cA$ can release $f_{\Delta}$ (for an appropriately chosen $\Delta$) via the Laplace mechanism. The 
error of $\cA$ on a graph $G$ is $|f(G)-f_{\Delta}(G)|$ plus the noise added for privacy, which is $\Theta(\Delta)$.
The Generalized Exponential Mechanism (GEM) 
\cite{RaskhodnikovaS16} can be 
used to select the value of $\Delta$ that approximately minimizes the expected error. 
The resulting algorithm has small error on graphs $G$ for which there is a small $\Delta$ such that $f_{\Delta}(G) \approx f(G)$. 

Lipschitz extensions can be defined for any metric space. A line of work in functional analysis seeks to understand general conditions under which, given a function $f$ and a subset of its domain, a Lipchitz extension to the entire domain exists~\cite{BenyaminiL98book}.
Such an extension always exists for real-valued $f$ \cite{McShane34}. However, it is not known to be efficiently computable even when $f$ is computable in polynomial time.
The main challenge for utilizing the framework of Lipschitz extensions lies in designing efficiently computable 
extensions that match the desired function on a large subset of the domain.

\subsubsection*{Our Lipschitz extensions for $\fcc$}\label{sec:intro-our-Lipschitz-extensions}
 Our Lipschitz extensions for the number of connected components are computable in polynomial time. Given a parameter $\Delta$ and a graph $G$, the value of our Lipschitz extension $f_{\Delta}(G)$ is obtained by solving a linear program over the forest polytope of $G$. 
 Let $S_\Delta$ be the anchor set of $f_\Delta$, i.e., the set of graphs $G$ with $f_\Delta(G)=\fcc(G).$ 
 \as{Then} 
 $S_\Delta$ contains all graphs that have a spanning forest of degree at most $\Delta$. This feature of our construction allows us to bound the error of the algorithm by the smallest possible maximum degree of a spanning forest of the graph.

We show that our construction of Lipschitz extensions is nearly optimal.
We consider two notions of optimality.
First, we compare our sets $S_\Delta$ to the largest possible anchor sets. To make optimal anchor sets well defined, we require them to be {\em monotone}, i.e., 
if $G$ is in the set, so are all of its induced subgraphs.
As it turns out, the largest monotone anchor set is unique and characterized in terms of the $\emph{down-sensitivity}$ of $f$, a quantity  introduced by the name of ``empirical global sensitivity'' by Chen and Zhou \cite{ChenZ13}.
The down-sensitivity of $f$ at a graph $G$, denoted $DS_f(G)$, measures the maximum change in the value of $f$ between any two node-neighboring induced subgraphs of $G$. Let $S^*_{\Delta} = \{G \mid DS_f(G) \leq \Delta\}$. Then $S^*_{\Delta}$ is the largest monotone anchor set for the function $f$ and parameter $\Delta$. 
A Lipschitz extension for the set $S^*_{\Delta}$ exists for all functions $f$ (it is defined and analyzed in \Lem{anchor_set}). In general, the evaluation of the extension takes exponential time even if $f$ is efficiently computable. 
We show that $S^*_{\Delta-1} \subseteq S_{\Delta}$, i.e., the anchor set for our Lipschitz extension with parameter $\Delta$ contains the largest monotone anchor set for  parameter $\Delta - 1$. When the Lipschitz extension is used in our 
algorithm, that difference of 1 barely changes the noise added for privacy.  Moreover, our construction runs in polynomial time.  

The second notion of optimality we consider is more general. It  was introduced by Cummings and Durfee \cite{CummingsD20}. They study the $\ell_\infty$ error of the extension, measured locally
over all induced subgraphs of any given 
graph. According to this measure, the error of our Lipschitz extension $f_\Delta$ on induced subgraphs of a graph $G$ is at most twice the  error of the best Lipschitz extension with parameter $\Delta -1$ on the same set of graphs. The best extension does not have to be efficiently computable and is selected separately for each $G$.

\subsection{Our Results}

We consider databases that represent
undirected, unweighted graphs $G$ with vertex set $V(G)$ and edge set $E(G)$. 
We start by formally defining the privacy notion of our algorithm.

\begin{definition}[Node-privacy] A randomized algorithm $\cA$ is $\eps$-node-private if for all node-neighboring graphs $G, G'$ and all events $S$ in the output space of $\cA$,
\begin{align*}
    \Pr[\cA(G) \in S] \leq e^\eps \Pr[\cA(G') \in S].
\end{align*}
\end{definition}

Our first idea is to restate the problem of estimating $\fcc$, the number of connected components, in terms of estimating the size of a spanning forest. Let $\fsf(G)$ denote the number of edges in a spanning forest of~$G$.
The reason $\fsf$ is an easier function to work with is that it is monotone (nondecreasing) under the addition of nodes and edges, whereas $\fcc$ can increase or decrease when we insert a new node with arbitrary edges.
The relationship between the two functions is straightforward:
\begin{align}\label{eq:cc-sf-connection}
\fcc(G) = |V(G)| - \fsf(G). 
\end{align}

The number of nodes in a graph can be easily estimated with an $\eps$-node-private algorithm with additive error $O(1/\eps)$ by using the standard Laplace mechanism.
Henceforth, we state our results in terms of approximating $\fsf$, with the understanding that it is equivalent to approximating $\fcc$ with additive error.  Our main result is a polynomial-time, $\eps$-node-private algorithm for approximating $\fsf$. Its error is upper bounded by $O(\Delta^* \ln \ln n)$ for constant $\eps$.

\begin{theorem}[Node-private algorithm for the size of the spanning forest] \label{thm:sf} 
 There exists an $\eps$-node-private algorithm $\cA$ that, given an $n$-node graph $G$ and a parameter $\eps>0$, runs in polynomial time. Let $\Delta^*$ be the smallest possible maximum degree of a spanning forest of $G$. If $E(G)\neq\emptyset$, then, with probability $1 -o(1)$, the output of the algorithm satisfies
\begin{align}\label{eq:error-bound}
    |\cA(G) - f_{\mathrm{sf}}(G)| \leq \Delta^* \cdot \widetilde{O}\Bparen{\frac{\ln \ln n}{\eps}}.
\end{align}
\end{theorem}

Note that $\Delta^*$ is at most
the maximum degree of $G$.
Therefore, the error bound in \Eqn{error-bound} also holds when $\Delta^*$ is replaced with the maximum degree of $G$. 

\subsubsection{Our Lipschitz extensions for $\fsf$}\label{sec:results-Lipschitz-ext} 
Our main technical contribution and the key tool in our algorithm design is a construction of a family of Lipschitz extensions for the size of the spanning forest. 
Recall from Section~\ref{sec:prelims-on-Lipschitz-extensions} that the first challenge in designing a family of Lipschitz extension is identifying their anchor sets. The anchor set of our
Lipschitz extension $f_{\Delta}$ contains the set of graphs that have a spanning forest of degree at most $\Delta$ (called a spanning $\Delta$-forest). One can easily see that  $\fsf$ is $\Delta$-Lipschitz on this set. 

The construction of $f_{\Delta}$ is based on the forest polytope of the input graph. The vertices of the polytope, represented by vectors $x \in \{0, 1\}^E$, correspond to forest subgraphs of $G$, where the weight of an edge is an indicator for whether the edge is in the forest. The forest polytope contains all convex combinations of the forest subgraphs of $G$. We consider the degree-bounded forest polytope of a graph, parameterized by $\Delta$. It introduces the additional requirement that the total weight of the edges incident on each vertex is at most $\Delta$.  Then $f_{\Delta}(G)$ equals the total maximum weight that can be assigned to the edges of $G$, while ensuring that the vector of weights is in the $\Delta$-bounded forest polytope. 
We show that $f_{\Delta}$ is $\Delta$-Lipschitz. It is also computable in polynomial time using an LP solver and an efficient linear separation oracle. 

Our extensions and their analysis are described in \Sec{sf}.

\subsubsection{Relation to Down-sensitivity} 

The \textit{down-sensitivity} of a function, defined next, provides a useful lens for thinking about instance-dependent guarantees~\cite{ChenZ13,RaskhodnikovaS16-E}.

\begin{definition}[Down-sensitivity~\cite{ChenZ13}]
For two graphs $H, H'$, we write $H \preceq H'$ (or $H \prec H'$) if $H$ is an induced (or proper induced) subgraph of $H'$.
The \emph{down-sensitivity} of $f\colon \cG \to \R$ at $G$, denoted by $DS_f(G)$, is defined as
\begin{align*}
    DS_f(G) = \max_{\substack{H \preceq H' \preceq G \\ H, H' \textnormal{ neighbors}  } } |f(H') - f(H)|.
\end{align*}
\end{definition}

The down-sensitivity of $\fsf$ and of $\fcc$ differ by at most 1 on all graphs (since their sum, $|V|$,  changes by at most 1 between node-neighboring graphs). 

Within the framework of Lipschitz extensions, the down-sensitivity characterizes the largest monotone anchor set that is possible for a $\Delta$-Lipschitz extension. \asnote{More generally, it is the tightest possible monotone lower bound on the error of any differentially private algorithm that approximates a given function $f$. Wish list: Write this is as a lemma.} It plays an important role in the accuracy guarantees in  previous work on node-private algorithms.
We show that the error of our algorithm for the size of the spanning forest (and the number of connected components) can be bounded in terms of the down-sensitivity (\Thm{ds} below). The theorem provides a comparable result to existing theorems on subgraph counts and degree distributions~\cite{ChenZ13,RaskhodnikovaS16-E}.

\begin{theorem}[Down-sensitivity guarantee of algorithm for the size of the spanning forest]\label{thm:ds}
\sr{If $E(G)\neq\emptyset$,}
with probability  $1 - o(1)$, the algorithm in \Thm{sf} satisfies
\begin{align*}
    |\cA(G) - f_{\mathrm{sf}}(G)| \leq DS_{\fsf}(G) \cdot \widetilde{O}\Bparen{\frac{\ln \ln n}{\eps}}.
\end{align*}
\end{theorem}
An analogous statement holds with ${\fcc}$ replacing ${\fsf}$ when we modify the algorithm according to \eqref{eq:cc-sf-connection}.

To prove \Thm{ds}, we establish a connection between the down-sensitivity of $\fsf$ at graph $G$ and the existence of a degree-bounded spanning forest of $G$. \Lem{trees} together with \Thm{sf} immediately imply \Thm{ds}.

\begin{lemma}\label{lem:trees}
Let $\Delta^*$ denote the smallest possible maximum degree of a spanning forest of a graph $G.$
Then $\Delta^* \leq DS_{\fsf}(G) + 1$.
\end{lemma}

 \Lem{trees} is obtained from two combinatorial results about spanning forests of graphs. The first, \Lem{induced-stars}, characterizes the down-sensitivity of $\fsf$ via the size of induced stars of~$G$. Given an integer $k \geq 1$ 
 and vertices $v_0, \dots, v_k\in V(G)$, we say that these vertices form an {\em induced $k$-star centered at $v_0$} in $G$ if $(v_0, v_i) \in E(G)$ for all $i \in [k]$ and $(v_i,v_j)\notin E(G)$ 
 for all $i,j \in [k].$
 
 \begin{lemma}\label{lem:induced-stars}
For a graph $G$, let $\inducedstar(G)$ denote the largest integer such that $G$ has an induced $\inducedstar(G)$-star. Then
 \begin{align*}
     DS_{\fsf}(G) = \inducedstar(G).
 \end{align*}
 \end{lemma}

Note that $s(G)$ can be much lower than the maximum degree of $G$ due to the requirement that the $s(G)$-star be induced. The second result, \Lem{stars}, states that a graph with no large induced stars has a low-degree spanning forest.
 
\begin{lemma}\label{lem:stars}
A graph with no induced $\Delta$-stars has a spanning $\Delta$-forest.
\end{lemma}

\Lem{stars} is the key link in the chain connecting the down-sensitivity of $\fsf$ to our algorithm's accuracy. 
The proof is constructive: we give a procedure that adds one vertex at a time to the spanning forest, modifying it at each step to maintain the degree bound. 

 
\paragraph{Down Sensitivity, Anchor Sets, and Our Extensions}~ As discussed in \Sec{intro-our-Lipschitz-extensions}, one way to measure the quality of a Lipschitz extension is in terms of its anchor set.
%
%
Given a function $f$ and parameter $\Delta$, the largest possible monotone anchor set of any $\Delta$-Lipschitz extension of $f$ is the set $S^*_{\Delta} = \{G \mid DS_f(G) \leq \Delta\}$. By ``largest'' we mean that every monotone subset of an anchor set is a subset of $S^*_{\Delta}$.
We give an explicit Lipschitz extension with anchor set $S^*_{\Delta}$ and a proof of the optimality of $S^*_{\Delta}$ in  \Sec{ds_general}. 
However, for general $f$, the construction need not be efficient, even if $f$ is computable in polynomial time.

The Lipschitz extensions $f_\Delta$ (for $\fsf$) used in our main algorithm are efficiently computable and have anchor sets that nearly match those of the extension based on down-sensitivity: their  anchor sets $S_\Delta$ contain all graphs with down-sensitivity at most $\Delta-1$, which means they contain the largest possible monotone anchor set for Lipschitz parameter $\Delta-1$.

\begin{lemma}[Nearly Optimal Anchor Sets]\label{lem:our-anchor-sets}
Let $G$ be a graph and  $\Delta \geq DS_{\fsf}(G) +1$. Then $f_\Delta(G)=\fsf(G)$. Consequently, for all $\Delta\geq 1$, 
\begin{align*}
S^*_{\Delta-1}\subseteq S_\Delta \, .
\end{align*}
\end{lemma}

We prove Lemmas~\ref{lem:trees}--\ref{lem:our-anchor-sets}  in \Sec{ds}.

\subsubsection{Optimality of our Lipschitz extensions in terms of $\ell_\infty$ error.}\label{sec:linfty-intro}

Next, we analyze
the optimality of our Lipschitz extension 
in the vein of results of Cummings and Durfee \cite{CummingsD20}. They design an algorithm for constructing a Lipschitz extension for general $f$. Their construction is 2-competitive with
the optimal extension with the same Lipschitz parameter. 
\sr{To measure the error of an extension $f_{\Delta}$, we}
define $\Err_G(f_{\Delta}, f)$ as $\max_{H \preceq G} |f_{\Delta}(H) - f(H)|$. Cummings and Durfee compare $\Err_G(f_{\Delta}, f)$ to $\Err_G(f^*, f)$ for all functions $f^*$ of bounded sensitivity. We give such a comparison for our Lipschitz extension for the size of the spanning forest.

\begin{definition}[Functions of bounded sensitivity]
Let $\mathcal{G}$ be the set of all undirected, unweighted graphs. For a Lipschitz parameter $\Delta > 0$, define 
\begin{align*}
    \mathcal{F}_{\Delta} = \{ f \colon \mathcal{G} \to \R \mid f \textnormal{ is $\Delta$-Lipschitz} \}. 
\end{align*}
\end{definition}

\begin{theorem}[Optimality of our Lipschitz extension]
\label{thm:optimality}
Let $\Delta \geq 1$ and $f_{\Delta}$ be our Lipschitz extension with parameter $\Delta$ for the size of the spanning forest. If $\Err_G(f_{\Delta}, \fsf) > 0$ then 
\begin{align}
   \Err_G(f_{\Delta}, \fsf) \leq \Big( 2 \cdot \min_{f^* \in \mathcal{F}_{\Delta-1}} \Err_G(f^*, \fsf) \Big) - 1. \label{eq:optimality}  
\end{align}
\end{theorem}
Cummings and Durfee show that for all functions $f$, their Lipschitz extension, denoted by $f_{\Delta}'$, satisfies $\Err_G(f_{\Delta}', f) \leq 2 \min_{f^* \in \mathcal{F}_{\Delta}} \Err_G(f^*, f)$,  Instead, in our result, the error of the extension is compared against the more restricted class of functions of sensitivity $\Delta-1$. However, our extension is computable in polynomial time, as opposed to the exponential construction of~\cite{CummingsD20}. 
The core of the proof is a lemma (\Lem{remove_set}) which, intuitively, explains the error of our Lipschitz extension $f_\Delta$ by attributing the error on a graph $G$ to one of its induced subgraphs. To prove it, we use a combinatorial result of Win \cite{Win89} on the decomposition of graphs with no spanning $\Delta$-forests. We prove \Lem{remove_set} and \Thm{optimality} in \Sec{optimality}.

\Thm{optimality} provides a strong type of optimality. In particular, it implies \Lem{our-anchor-sets} on the optimality of the anchor sets $S_\Delta$.

\subsubsection{Performance of our algorithm on specific graph families} 
We show that our algorithm provides a good approximation to $\fcc$ (the number of connected components) for two common graph models: the Erd\H{o}s-R\'enyi model $G(n, p)$ for subconstant $p$, and the geometric graph model. 

The $G(n, p)$ model generates a random $n$-node graph by independently adding an edge between each pair of nodes with probability $p$.
Erd\H{o}s and R\'enyi \cite{ErdosR60} showed that the connectivity of $G$ changes with $p.$
Consider the regime where 
$np =c$ for some constant $c > 0$. With probability $1-o(1)$, the graph will have $\fcc = \Omega(n)$ 
and maximum degree $O(\log n)$.
By \Thm{sf} and \Eqn{cc-sf-connection}, 
our node-private algorithm's estimate of $\fcc$ will then have additive error
$\pm\widetilde{O}((\log n)/\eps)$
and relative error $\tilde O(( \log^2 n)/\eps n)$.


A random geometric graph $G$ is defined by a set of $n$ vertices $V$ in the unit square, i.e., $V \subseteq [0, 1] \times [0, 1]$, and a distance parameter $r \in (0, 1)$. An edge exists between two vertices in $V$ if their Euclidean distance is at most $r$~\cite{penrose2003random}.  Geometric graphs have low-degree spanning forests, namely, of degree at most 6~\cite{balogh2011hamilton}.
We provide an alternative proof of this fact, via \Lem{stars}: 
a geometric graph does not have induced $6$-stars because 
it is impossible to fit $6$ points in the unit disk so that the distance between all points is strictly greater than $1$. By \Lem{stars}, a geometric graph has a spanning $6$-forest. By \Thm{sf} and \Eqn{cc-sf-connection}, our algorithm gives a
\sr{$\pm \as{\tilde O}((\ln \ln n)/\eps)$}
approximation to $\fcc$ in geometric graphs. Random geometric graphs and their variants have been used to model certain aspects of real-world networks,  including social graphs  \cite{bonato2015domination,flaxman2006geometric} and mobile networks~\cite{peres2013mobile,zinovyev2022space}.

\subsection{Prior Work on Private Graph Analysis}\label{sec:related_work}


Of the two natural adaptations of differential privacy to network data, edge-privacy is easier to achieve and (as a result) has been studied more extensively. Edge-private algorithms are known for a wide range of graph statistics and modeling tasks; see, for example, \cite{NissimRS07,HayLMJ09,BlockiBDS12,GuptaRU12,KarwaS12,Upadhyay13, WangWW13, WangWZX13,KarwaRSY14,ProserpioGM14,LuM14, ZhangCPSX15, MulleCB15, NguyenIR16, RoohiRT19,ZhangN19,AhmedLJ20,BlockiGM22,DLRSSY22}. The number of connected components has not been studied explicitly under edge-privacy, but it is easy to release with additive error $\Theta(1/\eps)$, since it can change by at most 1 with the insertion or removal of an edge.

Node privacy is a better fit for social network data, where one individual's data consists of all the relationships (edges) that the individual contributes to the graph. Its much stronger privacy protection makes it harder to attain.
Existing work addresses subgraph counts \cite{BlockiBDS13, KasiviswanathanNRS13, ChenZ13, 
DingZB018,  LiuML20}; degree and triangle distributions \cite{RaskhodnikovaS16,Dayll16,LiuML20}; parameter estimation in stochastic block models \cite{BorgsCS15, BorgsCSZ18, SealfonU21}; training of graph neural networks \cite{DaigavaneMSTAJ21}; and generating synthetic graphs \cite{ZhangNF20}. See \cite{RaskhodnikovaS16-E, MuellerUPRK22} for surveys on node-private algorithms and \cite{XiaCKHT021} for implementations of some of the algorithms.
\asnote{There are also papers that are not node-private because they assume (instead of enforcing) a sensitivity or degree bound: \cite{DingZB018} (see Proof of Lemma 1),\cite{SongLMVC18}, others?} 

Existing works on node-privacy that prove rigorous accuracy statements generally take one of two approaches. 
Some works assume that the input is generated by a distribution from a specific family and analyze how well their algorithm approximates the parameters of that distribution; that approach does not fit our setting of arbitrary fixed inputs.
%
Other works seek to formulate instance-dependent guarantees based on specific features of a graph, notably the maximum degree~\cite{KasiviswanathanNRS13,BlockiBDS13,RaskhodnikovaS16,Dayll16}. Chen and Zhou~\cite{ChenZ13} give a finer-grained analysis, showing that the accuracy of their algorithms for subgraph counts (and implicitly those of \cite{KasiviswanathanNRS13}) relates to the down-sensitivity of the input; that approach was subsequently refined and strengthened \cite{RaskhodnikovaS16,RaskhodnikovaS16-E}.

Our work goes further.\srnote{Is this a good place for this? Also, should we say which quantities are supermodular or have other simple structure described here?} We identify a feature of the graph---$\Delta^*$---that bounds the error of our algorithm and show that $\Delta^*$ can be bounded above in terms of the down-sensitivity of $\fsf$ (and by extension that of $\fcc$). Implementing this approach requires new ideas and combinatorial results, since $\fsf$ has a significantly different structure from that of subgraph counts: it is not a sum of quantities computable from the local view of each node, and it is not supermodular. Consequently, the connection to down-sensitivity is more subtle. 



\section{Preliminaries} \label{sec:prelims}

The most basic private mechanism for releasing a statistic $f$ returns the value of $f$ with additive noise scaled according to the global sensitivity of $f$. The noise follows a Laplace distribution. The Laplace distribution with mean $0$ and standard deviation $\sqrt{2}b$, denoted by $\mathrm{Lap}(b)$, has probability density $h(z) = \frac{e^{-|z|/b}}{2b}$. 

\begin{definition}[Global sensitivity \cite{DworkMNS16}]\label{def:GS} 
Given a function $f \colon \cG \to \R$, its global sensitivity, $GS_f$, is defined as 
\begin{align*}
    GS_f = \max_{\text{neighbors $G,G'$}} |f(G) - f(G')|.
\end{align*}
\end{definition}
\noindent Unless specified otherwise, we use $GS_f$ w.r.t.\ node-neighbors.

\begin{theorem}[Laplace Mechanism \cite{DworkMNS16}] 
\label{thm:laplace}
The algorithm $\cA$ that, given a graph $G$, outputs $\cA(G) = f(G) + \mathrm{Lap}(GS_f/\eps)$ is $\eps$-node-private.
\end{theorem}

\begin{lemma}[Tail of Laplace random variable] \label{lem:tail_laplace}
If $X \sim \mathrm{Lap}(b)$, then $\Pr[|X| \geq t \cdot b] = e^{-t}$.
\end{lemma}

Differential privacy is preserved under post-processing. Additionally, the outputs of multiple private algorithms can be combined to obtain an algorithm that has privacy protection linear in
the number of composed algorithms. 

\begin{lemma}[Composition and post-processing \cite{DworkMNS16,DworkKMMN06}]
\label{lem:composition}
If an algorithm $\cA$ runs 
$\eps$-node-private algorithms $\cA_1,...,\cA_t$
and applies a randomized algorithm $g$ to the outputs, then $\cA(G) = g(\cA_1(G), \dots, \cA_t(G))$ is $(t\eps)$-node-private.
\end{lemma}

\section{A Lipschitz Extension for the Size of the Spanning Forest}\label{sec:sf}

In this section, we prove \Thm{sf} by giving a polynomial-time $\eps$-node-private algorithm for approximating the size of the spanning forest of a graph $G$. 



We start by defining a family of Lipschitz extensions for $\fsf$ that are computable in polynomial time. Recall from \Sec{results-Lipschitz-ext} that our construction is based on the $\Delta$-bounded forest polytope of the input graph. 

\begin{definition}[Lipschitz extension for the size of the spanning forest] \label{def:sf}
    Given a vector $x \in \R^{E}$ and an edge $e \in E$, let $x(e)$ denote the value of the vector $x$ at edge $e$. For a subset $S \subseteq V$,  let $E[S]$ be the set of edges in the subgraph of $G$ induced by $S$. Let $\delta(v)$ denote the set of edges incident to $v$. For a set of edges $F \subseteq E$, denote by $x(F)$ the value $\sum_{e \in F} x(e)$.  Given $\Delta>0$, the $\Delta$-bounded forest polytope of $G$, denoted  $\cP_\Delta(G)$, consists of vectors $x \in \R^{E}$ that satisfy the following constraints:
\begin{align}
    x(e) &\geq 0  &\forall \: e \in E; \label{eq:zero}\\
    x(E[S]) &\leq |S| - 1  &\forall \: S \subseteq V, \:|S| \geq 2; \label{eq:tree} \\
    x(\delta(v)) &\leq \Delta  &\forall\: v \in V. \label{eq:degree}
\end{align}
The Lipschitz extension at $G$ with parameter $\Delta$ is defined as $f_\Delta(G) = \max_{x \in \cP_\Delta(G)} x(E)$. 
\end{definition}

Each $f_{\Delta}$ can be used to privately approximate $\fsf$ by outputting a private approximation of $f_{\Delta}$ via the Laplace mechanism. The Laplace noise is scaled according to $\Delta$, which is an upper bound on the global sensitivity of $f_{\Delta}$. We use the Generalized Exponential Mechanism to privately select $\hat{\Delta} \in [1, n]$, such that $f_{\hat{\Delta}}(G)$ has the lowest expected error when used to approximate $\fsf(G)$ via the Laplace mechanism. Finally, we output the private approximation of $f_{\hat{\Delta}}$ via the Laplace mechanism. This procedure is described in \Alg{sf}.
\begin{algorithm}[H]
  \setstretch{1.1}
  \caption{Node-Private Size of Spanning Forest} \label{alg:sf}
  \begin{algorithmic}[1]
    \Require{Graph $G$, privacy parameter $\eps > 0$, failure probability $\beta \in (0,1)$. }
    \State Run \Alg{gem} (based on GEM \cite{RaskhodnikovaS16}) with parameters $\eps/2$ and $\beta$ and access to \Alg{eval},  to obtain $\hat{\Delta}$. \label{step:gem}
    \State $f_{\hat{\Delta}}(G) \leftarrow \mathtt{EvalLipschitzExtension}(G, \hat{\Delta})$. 
    \State Return $f_{\hat{\Delta}}(G) + Z $ where $Z \sim \mathrm{Lap}(\frac{2\hat{\Delta}}{\eps})$. \label{step:laplace}
  \end{algorithmic}
\end{algorithm}

\begin{algorithm}[H]
  \setstretch{1.1}
  \caption{$\mathtt{EvalLipschitzExtension}$} \label{alg:eval}
  \begin{algorithmic}[1]
    \Require{Graph $G$, Lipschitz parameter $\Delta \in [1, n]$.}
    \State Solve the linear program from \Def{sf} to obtain $f_{\Delta}(G) \leftarrow  \max_{x \in \cP_\Delta(G)} x(E)$.
    \State Return $f_{\Delta}(G)$.
  \end{algorithmic}
\end{algorithm}

\Lem{extension} summarizes the properties of our family of Lipschitz extensions that allow us to use the Generalized Exponential Mechanism with this family.

\begin{definition}[Monotone in $\Delta$, Lipschitz Underestimates] \label{defn:extension}
Let $h \colon \cG \to \R$. The functions $\{ h_{\Delta} \}_{\Delta \in [1, \Delta_{\max}]}$ are a family of {\em monotone in $\Delta$, Lipschitz underestimates} for $h$ if: 
\begin{enumerate}
    \item \emph{\textsc{(Underestimation)}} $h_\Delta(G) \leq h(G)$ for all $\Delta \in [1, \Delta_{\max}]$ and all $G$.
    
    \item \emph{\textsc{(Monotonicity)}} $h_{\Delta_1}(G) \leq h_{\Delta_2}(G)$ for all $\Delta_1 < \Delta_2$ and all $G$. 
    
    \item \emph{\textsc{(Lipschitzness)}} $h_\Delta$ is $\Delta$-Lipschitz for all $\Delta \in [1, \Delta_{\max}]$. 
\end{enumerate}
\end{definition}

\begin{lemma}[Properties of the Lipschitz extension]
\label{lem:extension}
The Lipschitz extensions $\{f_{\Delta}\}_{\Delta \in [1, n]}$ in \Def{sf} are a family of monotone in $\Delta$, Lipschitz underestimates for $\fsf$. Moreover, for all $\Delta \in [1,n]$ and all $G$, the following hold.
\begin{enumerate}
    \item\label{item:delta-forest} If $G$ has a spanning $\Delta$-forest, then $f_\Delta(G) = f_{\mathrm{sf}}(G)$.
    \item \label{item:poly-time}$f_\Delta$ is computable in polynomial time.
\end{enumerate}
\end{lemma}

\begin{remark} The Lipschitz constant $\Delta$ for $f_{\Delta}$ is tight. To see this, consider the graph $G$ with $\Delta$ isolated vertices and the node-neighboring graph $G'$ obtained from $G$ by adding one vertex with edges to all vertices of $G$. Then $f_\Delta(G) = 0$ and $f_{\Delta}(G') = \Delta$. 
\end{remark}

\begin{proof}[Proof of \Lem{extension}]
We first show that $f_{\Delta}(G) \leq \fsf(G)$ for all $\Delta$ and $G$. Let $S_1, \dots, S_k \subseteq V(G)$ be the vertex sets of the connected components of $G$. For all $x \in \cP_\Delta(G)$ and for all $i \in [k],$ we have $x(E[S_i]) \leq |S_i| - 1$. Therefore, $x(E) \leq \sum_{i \in [k]}|S_i| - f_{\mathrm{cc}}(G) = n - f_{\mathrm{cc}}(G) = f_{\mathrm{sf}}(G)$. 

Monotonicity in $\Delta$ follows  from the fact that every vector $x \in \cP_\Delta(G)$ also satisfies $x \in \cP_{\Delta'}(G)$ for $\Delta' \geq \Delta$.

Next, we show that $f_{\Delta}$ is $\Delta$-Lipschitz. Let $G$ be obtained from $G'= (V', E')$ by removing a node $v$ and all its adjacent edges. Then $f_\Delta(G) \leq f_\Delta(G')$, since for each $x \in \cP_\Delta(G)$, there is a vector $x' \in \cP_\Delta(G')$ with $x(E) = x'(E')$. Namely, $x'$ has the same entries as $x$ in $E$ and has value $0$ in the entries $E'\setminus E$. 

 Let $x' \in \cP_\Delta(G')$ be such that $x(E')= f_\Delta(G')$. Consider the vector $x$ obtained from $x'$ by omitting all the entries pertaining to edges in $\delta(v)$. Then $x'(E') - x(E) \leq \Delta$, since $x'(\delta(v)) \leq \Delta$. Also, $x \in \cP_\Delta(G)$, and thus $x(E) \leq f_\Delta(G)$. We get 
\begin{align*}
   &|f_\Delta(G') - f_\Delta(G)| = f_\Delta(G') - f_\Delta(G) \\ 
   &= x'(E') - x(E) + x(E) - f_\Delta(G) \leq \Delta + 0 = \Delta.
\end{align*}
This  concludes the proof that $f_{\Delta}$ is $\Delta$-Lipschitz. 

To prove Item \ref{item:delta-forest}, let $F$ be the edges of a spanning $\Delta$-forest of $G$, i.e., $|\delta(v) \cap F|\leq \Delta$ for all vertices $v$ of $G$. Let $x_F$ be the vector whose values are $1$ for all $e \in F$ and $0$ otherwise.  Then  $x_F \in \mathcal{P}_\Delta(G)$ and $x_F(E) = f_{\mathrm{sf}}(G)$. Thus $f_{\Delta}(G) \geq x_F(E) = \fsf(G)$. By the underestimation property, we obtain $f_\Delta(G) = f_{\mathrm{sf}}(G)$.

Finally, we prove Item \ref{item:poly-time}. To compute $f_{\Delta}(G)$, we need to solve a linear program. Padberg and Wolsey~\cite{PadbergW83} show the existence of a polynomial-time separation oracle for condition \Eqn{tree} of the linear program that involves an exponential number of constraints. Conditions \Eqn{zero} and \Eqn{degree} can be clearly checked in polynomial-time. Thus, there exists a polynomial time algorithm for solving the linear program in \Def{sf} and computing $f_{\Delta}(G)$. 
\end{proof}

\subsection{From Extensions to the Main Algorithm}

To formulate and analyze our main algorithm (and prove \Thm{sf}),
we measure the ``quality'' of each extension $f_{\Delta}$
as an approximation for $\fsf(G)$ by its expected error when outputting $f_{\Delta}(G)$ privately via the Laplace mechanism, as in \Alg{sf}. We would like to select $\hat{\Delta} \in [1, n]$, so that $f_{\hat{\Delta}}$ from our family of Lipschitz extensions minimizes the expected error. 
The Generalized Exponential Mechanism selects an approximation to the true minimizer $\hat{\Delta}$ in a private way.
Consider any function $h\colon \cG \to \R$ and a family $\{ h_{\Delta} \}_{\Delta \in [1, \Delta_{\max}]}$ of monotone in $\Delta$, Lipschitz underestimates for $h$. The quality of each function $h_{\Delta}$ is measured by its approximation error, defined as
\begin{align}
    err_h(\Delta, G) := |h_{\Delta}(G) - h(G)| + \Delta/\eps . \label{eq:error}
\end{align}
 Note that $err_h(\Delta, G)$ is an upper bound for the expected error $\E[|h_{\Delta}(G) + \mathrm{Lap}(\Delta/\eps) - h(G)|]$ by the triangle inequality and the fact that $\E[|\mathrm{Lap}(b)|] = b$. The Generalized Exponential Mechanism outputs $\hat{\Delta} \in [1, \Delta_{\max}]$ that approximately minimizes the quantity $err_h(\Delta, G)$.

\begin{theorem}[GEM \cite{RaskhodnikovaS16}]
\label{thm:gem}
Fix $\eps >0$ and  $\beta \in (0, 1)$. 
Let $\{ h_{\Delta} \}_{\Delta \in [1, \Delta_{\max}]}$ be a family of monotone in $\Delta$, Lipschitz underestimates for $h\colon \cG \to \R$.
Then there exists an $\eps$-node-private algorithm (\Alg{gem}) obtained from the Generalized Exponential Mechanism, that outputs a value $\hat{\Delta}$ such that for all $\Delta \in [1, \Delta_{\max}]$ and all $G$, with probability at least $1-\beta$, it holds
\begin{align*}
    err_h(\hat{\Delta}, G) \leq err_h(\Delta, G) \cdot O\Big(\ln \frac{\ln (\Delta_{\max})}{\beta}\Big). 
\end{align*}
Furthermore, the algorithm runs in polynomial time if all $h_{\Delta}$ are polynomial time computable. 
\end{theorem}



We  now   we use \Thm{gem} together with \Lem{extension} to complete the proof of \Thm{sf}. 

\begin{proof}[Proof of \Thm{sf}]
Let $\cA$ be \Alg{sf} run with failure probability 
$\beta =\frac{1}{\ln \ln n}$. We first show that $\cA$ is $\eps$-node-private. \Step{gem} of algorithm $\cA$ is $(\eps/2)$-node-private by \Thm{gem}. \Step{laplace} is also $(\eps/2)$-node-private by \Thm{laplace}. By composition (\Lem{composition}),  algorithm $\cA$ is $\eps$-node-private. 

We now bound the error of $\cA$. With probability at least $1-\beta/2,$ we have\srnote{It is confusing what always holds and what holds w.h.p.\ in the derivation.}
\begin{align*}
    |\cA(G) - \fsf(G)| &= |f_{\hat{\Delta}}(G) + Z - \fsf(G)| \\ 
    &\leq |f_{\hat{\Delta}}(G) - \fsf(G)| + |Z| \\
    &\leq |f_{\hat{\Delta}}(G)  - \fsf(G)| + \frac{2\hat{\Delta}}{\eps}\ln\Big(\frac{2}{\beta}\Big).
\end{align*}
The last inequality follows from \Lem{tail_laplace}. By the definition in \Eqn{error}, we have $|\cA(G) - f_{\mathrm{sf}}(G)| \leq  err_{\fsf}(\hat{\Delta}, G) \cdot 2\ln(2/\beta)$ with probability at least $1-\beta/2$.

Let $\Delta^*$ be the smallest value in $[1, n]$ such that $G$ has a spanning $\Delta^*$-forest. 
By Item \ref{item:delta-forest} of \Lem{extension},  $|f_{\Delta^*}(G) - \fsf(G)| = 0$. Therefore, $err_{\fsf}(\Delta^*, G) = \frac{\Delta^*}{\eps}$. By \Lem{extension}, the functions $\{f_{\Delta}\}_{\Delta \in [1,n]}$ are a family of monotone in $\Delta$, Lipschitz underestimates for $\fsf$. Applying \Thm{gem} with $\Delta_{\max} = n$, we have that with probability at least $1 - \beta$, 
\begin{align*}
    |\cA(G) - \fsf(G)| \leq \frac{\Delta^*}{\eps} \cdot O\Big(\ln  \Big(\frac{\ln n}{\beta}\Big)\cdot \ln\frac 1\beta\Big).
\end{align*}
Since $\beta = \frac{1}{\ln \ln n}$, we obtain the desired result. Since the Lipschitz extensions are computable in polynomial time (Item \ref{item:poly-time} of \Lem{extension}), algorithm $\cA$ runs in polynomial time. This concludes the proof. 
\end{proof}



\section{Down-Sensitivity of the Size of the Spanning Forest} \label{sec:ds}

In this section, we prove \Lem{trees} which establishes a connection between the down-sensitivity of $\fsf$ and the existence of a bounded-degree spanning forest, and \Lem{our-anchor-sets} which connects down-sensitivity to the anchor sets of our extension. \Lem{trees} follows from two combinatorial results on spanning forests: \Lems{induced-stars}{stars}. We start by proving \Lem{stars}, which connects induced stars to the existence of bounded-degree spanning forests and is the key step in the proof of \Lem{trees}. 
The proofs of \Lems{trees}{induced-stars} are deferred to \Sec{ds_sub}. We prove \Lem{our-anchor-sets} in \Sec{our-anchor-sets}

\begin{figure}
\centering
\includegraphics[width = 0.42\textwidth]{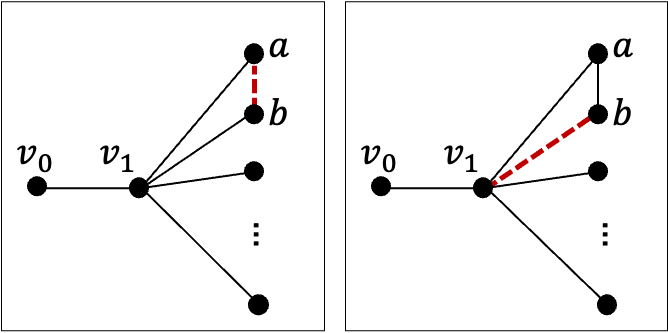}
\Description{The local repair operation.}
\caption{Before and after the local repair at vertex $v_1$. Black solid edges are
in the spanning forest. Dotted red edges are in the graph $G$, but not in the spanning forest.}
\label{fig:repair}
\end{figure}


We first give an overview of the main ideas of the proof of \Lem{stars}. 
The proof is by induction on the number of vertices in $G$. Suppose $G$ has no induced $\Delta$-stars, i.e., $\inducedstar(G) < \Delta$. 
Let $v_0$ be a vertex in $G$, which is not a cut vertex.
(A vertex is a {\em cut vertex} if its removal from the graph increases the number of connected components.)
Let $G'$ be the subgraph of $G$ induced by all vertices other than $v_0$. Then $s(G') \leq s(G) < \Delta$.
By the inductive hypothesis, $G'$ has a spanning forest $F$ of degree at most $\Delta$. Let $v_1$ be one of the neighbors of $v_0$ in $G$ and add the edge $(v_0, v_1)$ to $F$. Now $F$ is  a spanning forest of $G$, since $v_0$ is not a cut vertex.
However, the degree of $v_1$ in $F$ could now be $\Delta+1$. We 
modify $F$ to obtain a spanning forest of $G$ of degree at most $\Delta$ by performing a sequence of ``local repairs''. If $v_1$ has degree $\Delta+1$ in $F$, then  $v_1$ has two neighbors $a$ and $b$ in $F$, such that $a$ and $b$ are adjacent in $G$ (this is true since $G$ has no induced $\Delta$-stars). A local repair at $v_1$  replaces the edge $(v_1, a)$ in $F$ with the edge $(a, b)$. See \Fig{repair}. We show that this operation does not introduce a cycle, and $F$ is still a spanning forest of $G$. The degree of $v_1$ in $F$ is now $\Delta$, as desired, but the degree of $a$ could have increased from $\Delta$ to $\Delta +1$, thus calling for another local repair at $a$. At first glance, the local repair only pushes the problem around, rather than fixing it. However, with some care, we can show that the sequence of vertices $v_0, v_1, \dots$ where we perform the local repair forms a path in a spanning forest of $G$. Thus, the sequence of local repairs will eventually conclude, yielding a spanning forest of $G$ of degree at most $\Delta$. 

\begin{proof}[Proof of \Lem{stars}] 
Let $G$ be a graph on $n$ vertices. The proof is by induction on $n$. If $n = 1$, then $G$ has no induced $\Delta$-stars for all $\Delta > 0$. Graph $G$ also has a spanning $\Delta$-forest for all $\Delta > 0$. Thus the lemma holds. By the same reasoning, the lemma holds if $E(G) = \emptyset$. Suppose $E[G] \neq \emptyset$ and that the lemma holds for every graph with $n-1$ vertices. Then there exists some vertex $v_0 \in V(G)$ which is not isolated and not a cut vertex. E.g., consider a spanning forest of $G$ and let $v_0$ be one of the leaves of the spanning forest.  Let $G' = G \setminus \{v_0\}$. Then  $\inducedstar(G') \leq \inducedstar(G) < \Delta$. By the inductive hypothesis, there exists a spanning forest $F$ of $G'$ of degree at most $\Delta.$


Let vertex $v_1$ be a neighbor of $v_0$ in $G$ (recall that vertex $v_0$ is not isolated).  Let $F_0$ be the subgraph of $G$ consisting of $F$ and the edge $(v_0, v_1)$. Since $v_0$ is not a cut vertex, then $F_0$ is a spanning forest of $G$. If all vertices in $F_0$ have degree at most $\Delta$, then we have found the desired spanning forest. Otherwise, we perform a sequence of local repairs, as outlined in \Alg{trees}. The output of \Alg{trees} is a sequence $\{(v_i, F_i)\}$ of vertices and subgraphs of $G$. The subgraph $F_i$ is  obtained after a local repair at vertex $v_i$ (recall \Fig{repair}). In \Clm{trees}, we show that all subgraphs $F_i$ are spanning forests of $G$. Additionally, the vertices $v_0, v_1, \dots$ are distinct and thus the sequence $\{(v_i, F_i)\}$ is finite (it must end once all vertices of $G$ are output). By the description of \Alg{trees}, the final $F_i$ it outputs is a spanning forest of $G$ of degree at most $\Delta$, as desired. Applying \Clm{trees} thus concludes the proof.
\end{proof}


\begin{algorithm}
  \setstretch{1.1}
  \caption{Repair a $(\Delta+1)$-forest to get a spanning $\Delta$-forest} \label{alg:trees}
  \begin{algorithmic}[1]
    \Require{Graph $G$, a vertex $v_0$, and a forest $F_0$, defined above}    
    \Let{$i$}{1} 
    \While{max degree in $F_{i-1}$ is greater than $\Delta$}
         \State Let $v_i$ be a vertex with degree at least $\Delta+1$ in $F_{i-1}$. 
         \State Let $N$ be $\Delta$ neighbors of $v_{i}$ in $F_{i-1}$, with $v_{i-1} \notin N$.
         \State Let $a_i, b_i \in N$ such that $(a_i, b_i) \in E(G)$. \label{step:neighbors} \Comment{Since $\inducedstar(G) < \Delta$, such $a_i,b_i$ exist}.
         \State Let $F_i = (F_{i-1} \setminus \{(v_{i}, b_i)\}) \cup  \{(a_i, b_i)\}$.
         \State Output $(v_{i}, F_i)$. 
         \Let{$i$}{$i+1$} 
    \EndWhile
  \end{algorithmic}
\end{algorithm}

\begin{claim}\label{clm:trees}
For all $i \geq 0$ and pairs $(v_i, F_i)$ output by \Alg{trees}, the following hold:
\begin{enumerate}[(a)]
    \item $F_i$ is a spanning forest of $G$. 
    \item All vertices in $F_i$ have degree at most $\Delta+1$, and at most one of them has degree $\Delta+1$.
    \item $(v_i, v_{i+1})$ is an edge in both $F_i$ and $F_{i+1}$, assuming $F_i$ is not the final subgraph output by \Alg{trees}. 
    \item Vertices $v_0, \dots, v_i$ are distinct, and they form a path in~$F_i$. 
\end{enumerate}
\end{claim}
\begin{proof}
We prove Item (a) by induction on $i$. The claim is true for $i = 0$, since adding $(v_0, v_1)$ to $F$ does not introduce a cycle, and the number of connected components of $G$ and $G'$ is the same since $v_0$ is not a cut vertex. Suppose the claim holds for $F_{i-1}$. We show that it also holds for $F_i$. First note that $F_{i-1}$ and $F_i$ have the same number of edges. Therefore, it suffices to show that $F_i$ has no cycles. Suppose $F_i$ has a cycle. Then this cycle must include the edge $(a_i, b_i)$, whereas all other edges of the cycle are in $F_{i-1}$. Suppose the vertex $v_i$ is not included in the cycle. Then, by replacing the edge $(a_i, b_i)$ with the edges $(v_i, a_i)$ and $(v_{i}, b_i),$ we would obtain a cycle in $F_{i-1}$, a contradiction. Now suppose $v_i$ is included in the cycle. By replacing the edges $(v_i, a_i)$ and $(a_i, b_i)$ with the edge $(v_i, b_i),$ we would again obtain a cycle in $F_{i-1}$, a contradiction. Thus, $F_i$ is a spanning forest, which concludes the proof of Item (a). 

We also prove Item (b) by induction. The claim is true for $F_0$, since all vertices in $F_0$ have the same degree as in $F$, except for $v_1$, whose degree increases by 1 between $F$ and $F_0$. Since $\deg_{F}(v_1) \leq \Delta$, then $\deg_{F_0}(v_1) \leq \Delta+1$. Additionally, $v_0$ has degree exactly $1$, which is at most $\Delta$. 

Suppose the claim is true for $F_0, \dots, F_{i-1}$. If all vertices in $F_{i-1}$ have degree at most $\Delta$, this concludes the proof. Otherwise, there exists a unique vertex $v_i$ with $\deg_{F_{i-1}}(v_i) = \Delta + 1$.  Let  $a_i, b_i$ be the neighbors of $v_i$ defined in \Step{neighbors}.  Note that the only vertices whose degrees change between $F_i$ and $F_{i-1}$ are $a_i$ and $v_i$. The degree of $v_i$ decreases by 1, and thus $\deg_{F_i}(v_i) = \Delta$. The degree of $a_i$ increases by 1. By the inductive hypothesis, $\deg_{F_{i-1}}(a_i) \leq \Delta$, and thus $\deg_{F_i}(a_i) \leq \Delta + 1$. All other vertices in $F_i$ have the same degree as in $F_{i-1}$, thus their degree is at most $\Delta$. This concludes the proof of Item (b). 

We prove Item (c). Fix iteration $i$ and suppose $F_i$ is not the final subgraph output by \Alg{trees}. By Item (b), vertex $v_i$ has degree $\Delta+1$, whereas all other vertices  in $F_{i-1}$ have degree at most $\Delta$. Let $a_i, b_i$  be  the neighbors of $v_i$ defined in \Step{neighbors}. By the proof of Item (b), all vertices in $F_i$, except $a_i$, have degree at most $\Delta$.
Since $F_i$ is not the final forest, it must hold that $\deg_{F_i}(a_i) = \Delta + 1$. Therefore, $v_{i+1} = a_i$. Since $(v_i, a_i)$ is an edge in $F_i$, we obtain that $(v_i, v_{i+1})$ is an edge in $F_i$. Note that in iteration $i+1$, we pick neighbors $a_{i+1}$ and  $b_{i+1}$ of $v_{i+1}$ such that $a_{i+1} \neq v_i$ and $b_{i+1} \neq v_i$. Therefore, the edge $(v_i, v_{i+1})$ remains unchanged in $F_{i+1}$. This concludes the proof of Item (c). 

We prove Item (d) by induction. It clearly holds for $i=0$. Suppose it is true for iteration $i$. We prove it for iteration $i+1$. By Item (c), the edge $(v_i, v_{i+1})$ is in $F_i$. Note that $v_{i+1} \neq v_{i-1}$, since $v_{i+1} = a_i$, and we specifically choose $a_i \neq v_{i-1}$. Suppose $v_{i+1} = v_j$ for some $j < i-1$. By the inductive hypothesis, the path from $v_j$ to $v_i$ is in $F_i$. Together with the edge $(v_i, v_j)$, it forms a cycle in $F_i$, a contradiction. Therefore, $v_{i+1}$ is distinct from all vertices $v_0, \dots, v_i$. We now show that the vertices $v_0, \dots, v_{i+1}$ form a path in $F_{i+1}$. By a similar argument as for $v_{i+1} = a_i$, we can show that $b_i$ is distinct from all vertices $v_0, \dots, v_i$. Therefore, all edges on the path from $v_0$ to $v_i$ are also in the forest $F_{i+1}$. Since $(v_i, v_{i+1})$ is an edge in $F_{i+1}$, we obtain that the vertices $v_0,\dots, v_{i+1}$ form a path in $F_{i+1}$. 
\end{proof}

\subsection{Completing the Proof of \Lem{trees}} \label{sec:ds_sub}

In this section, we complete the proof of \Lem{trees}. To that end, we first prove \Lem{induced-stars} which characterizes the down-sensitivity of $\fsf$ via the size of induced stars of the graph. \Lem{trees} follows easily from \Lems{induced-stars}{stars}.

\begin{proof} [Proof of \Lem{induced-stars}]
We first show that $DS_{f_{\mathrm{sf}}}(G) \geq \inducedstar(G)$. 
Let $H'$ be a maximum induced star of $G$, i.e., $H'$ is an induced $\inducedstar(G)$-star, and let $v$ be its central vertex. Let $H = H' \setminus \{v\}$, i.e., $H$ consists of $\inducedstar(G)$ isolated vertices.  It follows that $f_{\mathrm{sf}}(H) = 0$ and $f_{\mathrm{sf}}(H') = \inducedstar(G)$. Since $H \preceq H' \preceq G$ and $H, H'$ are node-neighboring, then $D_{f_{\mathrm{sf}}}(G) \geq f_{\mathrm{sf}}(H') - f_{\mathrm{sf}}(H)= \inducedstar(G)$.

We now show that $DS_{f_{\mathrm{sf}}}(G) \leq \inducedstar(G)$. Let $H\prec H'$ be two node-neighboring induced subgraphs of $G$ that differ at node $v$ and satisfy $DS_{f_{\mathrm{sf}}}(G)= f_{\mathrm{sf}}(H') - f_{\mathrm{sf}}(H).$
If $\deg_{H'}(v) = 0$ then  $f_{\mathrm{sf}}(H') = f_{\mathrm{sf}}(H)$, and the lemma holds. 
Now assume $\deg_{H'}(v) \geq 1.$ Then
$f_{\mathrm{sf}}(H') \geq f_{\mathrm{sf}}(H)$, and in particular, $f_{\mathrm{sf}}(H') - f_{\mathrm{sf}}(H) 
= |V(H')|- f_{\mathrm{cc}}(H') - |V(H)|+ f_{\mathrm{cc}}(H)
= f_{\mathrm{cc}}(H) - f_{\mathrm{cc}}(H') + 1$. Let $S_v$ be the vertices of the connected component of $H'$ containing $v$, and let $S_1, \dots, S_k \subseteq S_v$ be the sets of vertices of the connected components of the subgraph induced by $S_v - \{v\}$. Note that $k \geq 1$, since $\deg_{H'}(v) \geq 1.$ Then $f_{\mathrm{cc}}(H) - f_{\mathrm{cc}}(H') = k-1$, and thus $DS_{f_{\mathrm{sf}}}(G)=f_{\mathrm{sf}}(H') - f_{\mathrm{sf}}(H)=k$. It remains to prove that $k \leq \inducedstar(G)$. 
For each $i \in [k]$, select some vertex $v_i\in S_i$ that is adjacent to $v$ in the graph $H'$. Then the set of vertices $\{v_i \: | \: i\in [k]\}$ is an independent set in $G$ because the vertices $v_i$ are in different connected components of $H$. As a result, the vertices $\{v_i \: | \: i\in [k]\} \cup \{v\}$ 
induce a star centered at $v$ in graph $G$. Therefore, $k \leq  \inducedstar(G)$, which concludes the proof.  
\end{proof}

\begin{proof}[Proof of \Lem{trees}]
   We show the following equivalent statement: If $DS_{\fsf}(G) \leq \Delta - 1$, then $G$ has a spanning $\Delta$-forest. Consider $G$ such that $DS_{\fsf}(G) \leq \Delta - 1$. By \Lem{induced-stars}, $s(G) = DS_{\fsf}(G) \leq \Delta - 1$. Thus, $G$ has no induced $\Delta$-stars, and by \Lem{stars} we obtain that $G$ has a spanning $\Delta$-forest, as desired. 
\end{proof}

\subsection{Anchor Sets of our Extensions}
\label{sec:our-anchor-sets}

In this section, we prove \Lem{our-anchor-sets}, which connects down-sensitivity to the anchor sets of our extensions. 

\begin{proof}[Proof of \Lem{our-anchor-sets}]
Fix $\Delta > 0$ and a graph $G$ such that $DS_{\fsf}(G) \leq \Delta - 1$. By \Lem{trees}, $G$ has a spanning $\Delta$-forest. Item \ref{item:delta-forest} of \Lem{extension} gives that $f_{\Delta}(G) = \fsf(G)$. Therefore, if $G \in S_{\Delta-1}^*$, then $G$ has a $\Delta$-forest, and as a result $G \in S_{\Delta}$. Thus $S_{\Delta-1}^* \subseteq S_{\Delta}$.  
\end{proof}

\section{Optimality of our Lipschitz Extensions}
\label{sec:optimality}

In this section, we prove \Thm{optimality}, which says that our Lipschitz extension for $f_{\mathrm{sf}}$, the size of the spanning forest, is close to optimal for this function.  
Our proof relies on the following combinatorial fact proved by Win \cite{Win89} about graphs with no spanning $\Delta$-forest. For a graph $G$ and vertex set $X \subseteq V(G)$, we denote by $G \setminus X$ the graph obtained from $G$ by removing all vertices in $X$ and their adjacent edges.

\begin{lemma}[Lemma 1, Theorem 1 in~\cite{Win89}] \label{lem:win}
Let $\Delta \geq 2$. If a graph $G$ has no spanning $\Delta$-forest, then there exists an induced subgraph $S \prec G$ and a vertex set $X \subset V(S)$, such that: 
\begin{enumerate}
    \item $S$ has a spanning $\Delta$-tree.
    \item $G$ has no edges between vertices in $G \setminus V(S)$ and vertices in $S \setminus X$.
    \item $f_{\mathrm{cc}}(S \setminus X) \geq |X|(\Delta - 2) + 2.$
\end{enumerate}
\end{lemma}

The key ingredient in our proof of \Thm{optimality} is the following lemma which, intuitively, explains the error of our Lipschitz extension $f_\Delta$ by attributing the error on a graph $G$ to one of its induced subgraphs. Recall from \Lem{extension} that $f_\Delta(G)\leq\fsf(G)$ for all graphs $G$ and that $f_\Delta$ can err only on graphs with no spanning $\Delta$-forest. \Lem{remove_set} shows that  every such graph $G$ has an induced subgraph $H$ such that $f_\Delta(G)$ is nearly as large as possible for any $\Delta$-Lipschitz lower bound on $\fsf(G)$ because $f_{\Delta}(G)$ cannot deviate too far from $\fsf(H)$.


\begin{lemma}\label{lem:remove_set}
Let $\Delta \geq 1$. If a graph $G$ has no spanning $\Delta$-forest, then there exists a proper induced subgraph $H \prec G$ such that 
\begin{align}\label{eq:remove_set}
    f_{\Delta}(G) \geq \fsf(H) + (\Delta - 1)d(G, H) + 1.
\end{align}
\end{lemma}

\begin{proof} Let $G=(V,E)$ be a graph with no spanning $\Delta$-forest.

To see that \Eqn{remove_set} holds for $\Delta = 1$, let $H$ be a subgraph consisting of a single vertex of $G$. Then $\fsf(H)=0,$ and the right-hand size of \Eqn{remove_set} evaluates to 1. Since $G$ has no spanning $\Delta$-forest, $G$ has at least one edge. Thus, $f_{\Delta}(G) \geq 1$, and \Eqn{remove_set} holds for $\Delta=1$. 

Now fix some $\Delta \geq 2$. The proof is by induction on $|V| + |E|$, the total number of vertices and edges of $G$. Since $G$ has no spanning $\Delta$-forest, it must have a vertex of degree at least $\Delta+1$. 
The smallest such graph (in terms of $|V|+|E|$) is
a $(\Delta+1)$-star. 

For the base case, let $G$ be a $(\Delta+1)$-star. Denote its central vertex by $v$. Let $H = G \setminus \{v\}$. Then  $f_{\Delta}(G) = \Delta$, $\fsf(H) = 0$, and $(\Delta - 1)d(G, H) + 1 = \Delta$, since $d(G, H) = 1$. Thus, the claim holds for $G$ (and is tight). 

For the inductive step, consider a graph $G$ that contains a $(\Delta+1)$-star as a proper subgraph. See \Fig{remove_set} for an illustration of the subgraphs of $G$ considered in the rest of the proof.

Suppose \Eqn{remove_set} holds for all proper subgraphs of $G$.  Let $S \preceq G$ and $X \subset V(S)$ be the subgraph and the vertex set given by \Lem{win}. Then
\begin{align}\label{eq:f-delta}
    f_{\Delta}(G) \geq   f_{\Delta}(S) +   f_{\Delta}(G \setminus V(S)) = \fsf(S) + f_{\Delta}(G \setminus V(S)), 
\end{align}
where the equality holds by Item \ref{item:delta-forest} of \Lem{extension}, since $S$ has a spanning $\Delta$-tree. We apply Item 1 of \Lem{win} and \Eqn{cc-sf-connection}: 
\begin{align}
    \fsf(S) &- \fsf(S \setminus X) \nonumber\\
    &= |V(S)| - 1 - |V(S \setminus X)| + f_{\mathrm{cc}}(S \setminus X) \nonumber\\
    &\geq |X| - 1 + |X|(\Delta - 2) + 2 = |X|(\Delta-1) + 1, 
    \label{eq:X1-bound}
\end{align}
where the inequality holds by Item 3 of \Lem{win}. Combining \Eqn{f-delta} and the lower bound on $\fsf(S)$ implied by \Eqn{X1-bound} yields
\begin{align}
    f_{\Delta}(G) \geq \fsf(S \setminus X) + |X|(\Delta-1) + 1 +  f_{\Delta}(G \setminus V(S)). \label{eq:pre-final}
\end{align}

First, consider the case when $G \setminus V(S)$ has a spanning $\Delta$-forest. By Item \ref{item:delta-forest} of \Lem{extension}, $f_{\Delta}(G \setminus V(S)) = \fsf(G \setminus V(S))$. Let $H = G \setminus X$. By \Lem{win}, there are no edges between $G \setminus V(S)$ and $S \setminus X$. Thus, $\fsf(S \setminus X) + \fsf(G \setminus V(S)) = \fsf(G \setminus  X) = \fsf(H)$. Also, $|X| =d(G, H)$. Substituting these two equalities into \Eqn{pre-final}, we obtain that \Eqn{remove_set} holds in this case. 

Finally, we consider the case when $G \setminus V(S)$ has no spanning $\Delta$-forest. Let $G_1=G \setminus V(S)$. 
By Item 3 of \Lem{win}, $V(S)\neq\emptyset,$ so
$G_1$ is a proper subgraph of $G$. By the inductive hypothesis, there exists an induced subgraph $H_1 \prec G_1$ such that
\begin{align*}
   f_{\Delta}(G\setminus V(S))= f_{\Delta}(G_1) \geq \fsf(H_1) + d(G_1, H_1)(\Delta - 1) + 1.
\end{align*}
Plugging this bound on $f_\Delta(G\setminus V(S))$ into \Eqn{pre-final}, we get
\begin{align*}
    f_{\Delta}(G) \geq& \fsf(S \setminus X) + \fsf(H_1) \\
    &+ (|X| + d(G_1, H_1))(\Delta - 1) + 2. 
\end{align*}
Let $H' = (S \setminus X) \cup H_1$. By Item 2 of \Lem{win}, there are no edges between $G_1$ and $S \setminus X$. Hence, $\fsf(S \setminus X) + \fsf(H_1) = \fsf(H')$. Also note that $|X| + d(G_1, H_1) = d(G, H')$. We obtain
\begin{align*}
    f_{\Delta}(G) \geq \fsf(H') + (\Delta - 1)d(G, H') + 2. 
\end{align*}
Thus, \Eqn{remove_set} holds in this case, too, concluding the proof. 
\end{proof}

\begin{figure}
\centering
\includegraphics[width = 0.3\textwidth]{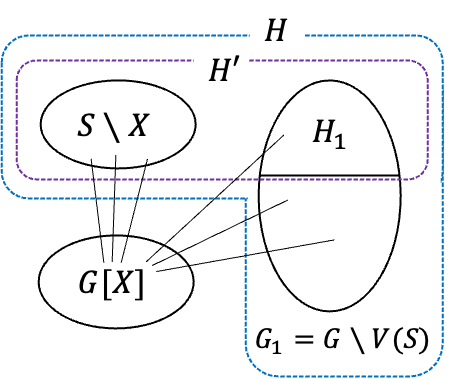}
\Description{The local repair operation.}
\caption{The subgraphs in \Lem{remove_set}. Graph $G[X]$ is the subgraph of $G$ induced by $X$. There are no edges between $S\setminus X$ and $G_1 = G\setminus V(S)$. Graph $H_1$ is an induced subgraph of $G_1$. Finally, $H = G \setminus X$ and $H' = H_1 \cup (S \setminus X)$. }
\label{fig:remove_set}
\end{figure}

We now prove \Thm{optimality} using \Lem{remove_set}. Except for \Lem{remove_set}, which is at the core of our argument, the rest of the proof builds on an argument of Cummings and Durfee~\cite{CummingsD20}.
\begin{proof}[Proof of \Thm{optimality}]
Let $G=(V,E)$ be a graph on which our Lipschitz extension errs, i.e., $\Err_G(f_\Delta, \fsf) > 0$.
The proof proceeds by induction on $|V| + |E|$, the total number of vertices and edges of $G$. By Item \ref{item:delta-forest} of \Lem{extension}, since $\Err_G(f_\Delta, \fsf) > 0$, graph $G$ has no spanning $\Delta$-forest. Therefore, we can consider the same base case as in the proof of \Lem{remove_set}.

For the base case, let $G$ be a $(\Delta+1)$-star. Let $v$ be its central vertex. For this graph, the left-hand side of \Eqn{optimality} equals 1. Let $f^* \in \cF_{\Delta - 1}$. By definition, $|f^*(G) - f^*(G \setminus \{v\})| \leq \Delta - 1$. In contrast, $\fsf(G) = \Delta + 1$ and $\fsf(G \setminus \{v\}) = 0$. The quantity $\max \{ |f^*(G) - \fsf(G)|, |f^*(G \setminus \{v\}) - \fsf(G \setminus \{v\}  |\}$ is minimized when $f^*(G) = \Delta$ and $f^*(G \setminus \{v\}) = 1$ (under the Lipschitz requirement on $f^*$). Thus, the right-hand side of \Eqn{optimality} equals 1, and  the theorem holds when $G$ is a $(\Delta+1)$-star.

For the inductive step, consider a graph $G$ that contains a $(\Delta+1)$-star as a proper subgraph and $\Err_G(f_{\Delta}, \fsf) > 0$. Suppose the theorem holds for all proper subgraphs of $G$. We show that it also holds for $G$. Let $\widetilde{H} \preceq G$ be the induced subgraph of $G$ that maximizes the quantity $|f_{\Delta}(H) - \fsf(H)|$ over all induced subgraphs $H\preceq G$. Then $\Err_{\widetilde{H}}(f_{\Delta}, \fsf) > 0$.

First, consider the case when
$\widetilde{H} \neq G$, i.e., $\widetilde{H} \prec G$. 
Then
\begin{align} 
     \Err_G&(f_{\Delta}, \fsf) 
     = \max_{H \preceq G} |f_{\Delta}(H) - \fsf(H) | \nonumber\\
     &= \max_{H \preceq \widetilde{H}} |f_{\Delta}(H) - \fsf(H) | \label{eq:err_G1}\\
     &\leq 2 \min_{f^* \in \cF_{\Delta-1}(\cG)} \max_{H \preceq \widetilde{H}} | f^*(H) - \fsf(H) |  - 1 \label{eq:err_G2}\\  
     &\leq 2 \min_{f^* \in \cF_{\Delta-1}(\cG)} \max_{H \preceq G} | f^*(H) - \fsf(H)  |  - 1,  \label{eq:err_G3}
\end{align}
where \Eqn{err_G1} follows by our choice of $\widetilde{H}$, the  inequality in \Eqn{err_G2} holds by the inductive hypothesis, and \Eqn{err_G3} holds since all induced  subgraphs of $\widetilde{H}$ are also induced subgraphs of $G$. Thus, the theorem holds for $G$. 

Finally, suppose $\widetilde{H} = G$, i.e., $\Err_G(f_{\Delta}, \fsf) = |f_{\Delta}(G) - \fsf(G)|$. Then $\Err_G(f_{\Delta}, \fsf) =\fsf(G)-f_\Delta(G)$ by the underestimation property of our extension (\Lem{extension}). Moreover, since $f_{\Delta}(G) \neq \fsf(G)$, Item \ref{item:delta-forest} of \Lem{extension} implies that $G$ has no spanning $\Delta$-forest. Thus, we can apply \Lem{remove_set}.
It gives us that $G$ has an induced subgraph $H \preceq G$ 
such that $f_{\Delta}(G) \geq \fsf(H) + (\Delta -1)d(G, H) + 1$. For all $f^* \in \cF_{\Delta-1},$ we have that $|f^*(G) - f^*(H)| \leq (\Delta-1) \cdot d(G, H)$. Thus,
\begin{align*}
    f_{\Delta}(G) &\geq \fsf(H) + |f^*(G) - f^*(H)| + 1 \\
    &\geq \fsf(H) + f^*(G) - f^*(H) + 1.
    \end{align*}
We use this inequality to get a bound on the error:    
    \begin{align*}
    \Err_G(f_{\Delta}, \fsf)
    &= \fsf(G) - f_{\Delta} (G) \\
    &\leq \fsf(G)- \fsf(H) - f^*(G) + f^*(H) - 1\\
    &\leq |\fsf(G) - f^*(G)| + |f^*(H) - \fsf(H)| - 1 \\
    & \leq 2\cdot \Err_G(f^*, \fsf)  - 1.
\end{align*}
Since this holds for all $f^* \in \cF_{\Delta-1},$
we finally obtain that $\Err_G(f_{\Delta}, \fsf) \leq 2 \min_{f^* \in \cF_{\Delta-1}
} \Err_G(f^*, \fsf)  - 1.$
\end{proof}

\section{Conclusion}
\label{sec:concl}

In this work, we presented the first node-differentially private algorithm for releasing the number of connected components of an undirected graph, a foundational query in graph databases. Our algorithm provides instance-based accuracy guarantees that are parameterized by the smallest possible maximum degree of a spanning forest of the input graph and relies on an efficiently computable family of Lipschitz extensions.  The functions in this family can be evaluated in polynomial time by solving an LP with an exponentially large number of constraints  using the ellipsoid algorithm with an efficient separation oracle. An interesting open direction is designing a faster algorithm with the same privacy and utility guarantees.

\bibliographystyle{alpha}
\bibliography{references}


\appendix

\section{Down-Sensitivity Error Guarantees for Approximating General Functions} \label{sec:ds_general}

In this section, we prove two results related to down-sensitivity. The first, \Thm{ds_general}, is a result of Raskhodnikova and Smith \cite{RaskhodnikovaS16-E} which states that for every monotone nondecreasing function $f$, there exists a node-private algorithm for approximating $f$ whose error is bounded by the down-sensitivity of $f$. A function $f \colon \cG \to \R$ is monotone nondecreasing if $f(H) \leq f(H')$ for all $H \preceq H'$. Then, in \Lem{anchor_set}, we show that the largest monotone anchor set  for a function $f$ and Lipschitz constant $\Delta$ is characterised by the down-sensitivity of $f$.

We start by proving \Lem{ds_extension}, which is used in the proof of \Thm{ds_general}.

\begin{lemma}[Lipschitz extension based on down-sensitivity] \label{lem:ds_extension}
Let $f\colon \mathcal{G} \to \R$ be a monotone nondecreasing function.  Given $\Delta > 0$, define the family of functions
\begin{align*}
    \widehat{f}_{\Delta}(G) = \min_{\substack{H \preceq G \\ DS_f(H) \leq \Delta } } \Big( f(H) + \Delta \cdot d(H, G) \Big).
\end{align*}
The functions $\widehat{f}_{\Delta}(G)$ are a family of monotone in $\Delta$, Lipschitz underestimates for $f$. Moreover, if $DS_f(G) \leq \Delta$, then $\widehat{f}_{\Delta}(G) = f(G)$. 
\end{lemma}
\begin{proof}

To show that the functions underestimate $f$, note that $\widehat{f}_{\Delta}(G) \leq f(G) + \Delta \cdot d(G, G) = f(G)$. 
To prove monotonicity in $\Delta$, fix $\Delta_1, \Delta_2$ and let $H \preceq G$ such that $DS_f(H) \leq \Delta_2$. Consider the empty graph $\emptyset$. Then $f(H) + \Delta_2 \cdot  d(H, G) \geq f(\emptyset) + \Delta_1 d(\emptyset, G) \geq \widehat{f}_{\Delta_1}(G)$, since $\emptyset \preceq H$ and $f$ is monotone nondecreasing. Since the inequality holds for all $H$ with $DS_f(H) \leq \Delta_2$, we obtain that $\widehat{f}_{\Delta_2}(G) \geq  \widehat{f}_{\Delta_1}(G)$.  

We now show $\Delta$-Lipschitzness. Let $G, G'$ be two neighboring graphs so that $V(G') = V(G) \cup \{v\}$, for $v \notin V$. We show that $\widehat{f}_{\Delta}(G) \leq \widehat{f}_{\Delta}(G') \leq \widehat{f}_{\Delta}(G) + \Delta$, which implies $|\widehat{f}_{\Delta}(G') - \widehat{f}_{\Delta}(G)| \leq \Delta$. 

We start by showing that $\widehat{f}_{\Delta}(G) \leq \widehat{f}_{\Delta}(G')$. Let $H' \preceq G'$ be the induced subgraph of $G'$ such that $DS_f(H') \leq \Delta$ and $\widehat{f}_{\Delta}(G') = f(H') + \Delta \cdot d(H', G')$. Suppose that $v \notin V(H')$. Then, $H$ is an induced subgraph of $G$ and $d(H', G) = d(H', G') - 1$. It follows that $\widehat{f}_{\Delta}(G) \leq f(H') + \Delta \cdot d(H', G) = f(H') + \Delta \cdot (d(H', G') - 1) = \widehat{f}_{\Delta}(G') - \Delta$. Suppose now that $v \in V(H')$. Let $H = H' \setminus \{v\}$. Then $H$ is an induced subgraph of $G$ and $d(H', G') = d(H, G)$. Additionally, $DS_f(H) \leq DS_f(H') \leq \Delta$. Since $f$ is monotone nondecreasing, it follows that $\widehat{f}_{\Delta}(G) \leq f(H) + \Delta \cdot d(H, G) \leq f(H') + \Delta \cdot d(H', G') = \widehat{f}_{\Delta}(G')$. In either case, we conclude that $\widehat{f}_{\Delta}(G) \leq \widehat{f}_{\Delta}(G')$. 

We now show that $\widehat{f}_{\Delta}(G') \leq \widehat{f}_{\Delta}(G) + \Delta$. Let $H \preceq G$ so that $DS_f(H) \leq \Delta$ and $\widehat{f}_{\Delta}(G) = f(H) + \Delta \cdot d(H, G)$. Since $H \preceq G'$, then $\widehat{f}_{\Delta}(G') \leq f(H) + \Delta \cdot d(H, G') = f(H) + \Delta \cdot (d(H, G) + 1) = \widehat{f}_{\Delta}(G) + \Delta$. 
This concludes the proof of Lipschitzness. 

Finally, we show that if $DS_f(G) \leq \Delta$ then $\widehat{f}_{\Delta}(G) = f(G)$. Suppose $DS_f(G) \leq \Delta$. Then $f$ is $\Delta$-Lipschitz for the set of induced subgraphs of $G$. Together with the fact that $f$ is monotone nondecreasing, we obtain $f(G) - f(H)  \leq \Delta \cdot d(H, G)$ for all induced subgraphs $H$ of $G$. Therefore $\widehat{f}_{\Delta}(G) = f(G)$. 
\end{proof}

We now use \Lem{ds_extension} to prove \Thm{ds_general}. 

\begin{theorem}[Theorem 2 of \cite{RaskhodnikovaS16-E}]
\label{thm:ds_general}
Let $f\colon \mathcal{G} \to \R$ be a monotone nondecreasing function. Let $\cG^{(n)}$ denote the set of all $n$-node graphs and suppose that $\max_{G' \in \cG^{(n)}} DS_f(G) \to \infty$ and $n\to \infty$. There exists an $\eps$-node-private algorithm $\cA_f$ that given an $n$-node graph $G$ and privacy parameter $\eps > 0$, with probability $1-o(1)$ satisfies,
\begin{align*}
   |\cA_f(G) - f(G)| 
    \leq \frac{DS_f(G)+1}{\eps} \cdot \widetilde{O}\Big(\ln \ln \max_{G' \in \cG^{(n)}} DS_f(G') \Big),
\end{align*}
\end{theorem}
\begin{proof}

Consider the family of functions defined in \Lem{ds_extension}, which are  monotone in $\Delta$, Lipschitz underestimates of $f$. Let $\beta = \frac{1}{\ln \ln \max_{G'} DS_f(G')}$ denote the failure probability. Note that $\beta = o(1)$. Let $\cA_f$ be the algorithm that runs the algorithm of \Thm{gem} with this family of Lipschitz extensions and parameters $\eps/2, \beta/2$ to select a threshold $\hat{\Delta} \in [1, \max_{G'}DS_f(G')]$. It then outputs $f_{\hat{\Delta}}(G) + \mathrm{Lap}(\hat{\Delta}/(2\eps))$. By composition (\Lem{composition}), $\cA_f$ is $\eps$-node-private. 

We analyze the error of $\cA_f$. By a similar argument as in the proof of \Thm{sf}, we have that with probability at least $1 - \beta/2$, it holds 
\begin{align*}
    |\cA_f(G) - f(G)| \leq  err_{\fsf}(\hat{\Delta}, G) \cdot 2\ln(2/\beta).
\end{align*}
Let $\widetilde{\Delta} = DS_f(G) + 1$ (note that $\widetilde{\Delta} \geq 1)$. By \Lem{ds_extension}, $err(\widetilde{\Delta}) = \frac{DS_f(G) + 1}{\eps}$. Applying \Thm{gem} with our choice of $\beta$ concludes the proof.  
\end{proof}

Finally, we prove \Lem{anchor_set}. 

\begin{lemma}\label{lem:anchor_set}
Given a function $f \colon \mathcal{G} \to \R$, a constant $\Delta > 0$, and a $\Delta$-Lipschitz function $\hat{f} \colon \mathcal{G} \to \R$, 
let $S_{\hat{f}} = \{ G \mid \hat{f}(G) = f(G) \}$. Let $S_{\hat{f}}'$ denote the largest monotone subset of $S_{\hat{f}}$ 
Then $S_{\hat{f}}' \subseteq \{ G \mid DS_f(G) \leq \Delta \}$. 
\end{lemma}
\begin{proof}
Suppose $G \in S_{\hat{f}}'$. Since $S_{\hat{f}}'$ is monotone, all induced subgraphs of $G$ are also in $S_{\hat{f}}'$. Then, 
\begin{align}\label{eq:Lip}
    \Delta &\geq \max_{\substack{H \preceq H' \preceq G \\ H, H' \textnormal{neighbors}  } } |\hat{f}(H) - \hat{f}(H')| \\\label{eq:ds}
    &=   \max_{\substack{H \preceq H' \preceq G \\ H, H' \textnormal{neighbors}  } } |f(H) - f(H')| = DS_f(G).
\end{align}
The inequality in (\ref{eq:Lip}) holds since $\hat{f}$ is $\Delta$-Lipschitz. The first equality in (\ref{eq:ds}) holds since $H, H' \in S_{\hat{f}}' \subseteq S_{\hat{f}}$ for all induced subgraphs of $G$. The second equality in (\ref{eq:ds}) is the definition of down-sensitivity. Therefore, $DS_f(G)\leq\Delta$ for all $G \in S_{\hat{f}}'$.
\end{proof}

\section{GEM}\label{sec:GEM}

    In this section, we state \Alg{gem} which applies the Generalized Exponential Mechanism of Raskhodnikova and Smith \cite{RaskhodnikovaS16} to the task of threshold selection for a family of Lipschitz extensions. This completes the description of \Alg{sf}. Given a (possibly infinite) family of Lipschitz extensions, \Alg{gem} first obtains finitely many score functions $q_i, i\in I$ of bounded sensitivity, defined in terms of the Lipschitz extensions\footnote{The functions $q_i$ defined in \Step{define_q} of \Alg{gem} can actually have high sensitivity, since $h$ can have high sensitivity. However, we only consider Lipschitz extensions that underestimate the true function $h$ (see \Lem{extension} and \Thm{gem}). 
In this case, minimizing $q_i$ is equivalent to minimizing $h(G) - h_i(G) + \frac{i}{\eps}$, and since the minimization is over $i$, we can treat $h(G)$ as a constant. So we can equivalently define the functions $q_i$ as $-h_i(G) + \frac{i}{\eps}$, which indeed have low sensitivity.}.
It then uses GEM to select the optimal score function.

GEM is a generalization of, and also includes as a subroutine, the Exponential Mechanism of McSherry and Talwar \cite{McSherryT07}. In the node-privacy setting, both algorithms take as input a graph $G$ and score functions $q_i \colon \cG \to \R, i \in I$ of bounded sensitivity. The goal of both algorithms is to output an index $\hat{\imath} \in I$ such that $q_{\hat{\imath}}$ approximately minimizes the value on $G$ amongst all score functions, i.e., $q_{\hat{\imath}}(G) \approx min_{i \in I} q_i(G)$. The Exponential Mechanism assumes a common upper bound $\Delta$ on the sensitivities of the score functions. GEM uses the fact that the score functions might have different sensitivities, and provides a better utility guarantee in the case when the optimal score function has a much lower sensitivity than the maximum of all sensitivities.

\begin{algorithm}[H]
  \setstretch{1.3}
  \caption{Threshold Selection for Lipschitz Extensions via GEM, adapted for node-privacy} \label{alg:gem}
  \begin{algorithmic}[1]
    \Require{Function $h \colon \cG \to \R$, access to subroutine $\mathtt{EvalLipschitzExtension}$ for evaluating the family of Lipschitz extensions $\{ h_{\Delta}\}_{\Delta \in [1, \Delta_{\max}]}$,}  privacy parameter $\eps > 0$, failure probability $\beta \in (0,1)$
    \State Let $k = \lfloor \log_2(\Delta_{\max}) \rfloor$, $I = \{2^0, \dots, 2^k\}$, and $t =  \frac{2 \log (k / \beta)}{\eps}$. 
    \State \textbf{for} $i \in I$ \textbf{do}:
    \State \indent $h_i(G) \leftarrow \mathtt{EvalLipschitzExtension}(G, i)$.
    \State \indent $q_i(G) \leftarrow |h_{i}(G) - h(G)| + \frac i{\eps}$. \label{step:define_q} 
    \State \textbf{for} $i \in I$ \textbf{do}:
    \State \indent  $s_i(G) \leftarrow \max_{j \in I} \dfrac{(q_i(G) + ti) - (q_j(G) + tj) }{i + j}$.
    \State To obtain $\hat{\imath}$, run the Exponential Mechanism~\cite{McSherryT07} with graph $G$, score functions $\{s_i :i \in I\}$,  and privacy parameter $\eps$. 
    \State Return $\hat{\imath}$. 
  \end{algorithmic}
\end{algorithm}

In \Thm{exp}, we define the Exponential Mechanism and state its privacy guarantee. \Alg{gem} describes GEM 
tailored to the specific task of selecting an optimal $\Delta$ from a (possibly infinite) family of Lipschitz extensions $\{ h_{\Delta}\}_{\Delta \in [1, \Delta_{\max}]}$. The utility, privacy, and running time guarantees of \Alg{gem} are stated in \Thm{gem}.

\begin{theorem}[Exponential Mechanism \cite{McSherryT07}] \label{thm:exp}
Given a graph $G$, finitely many score functions $q_i \colon \cG \to \R, i \in I$
with global sensitivity at most $\Delta$, and a privacy parameter $\eps > 0$, the Exponential Mechanism $\cA$  
outputs an index $\hat{\imath} \in I$ drawn from the distribution $\Pr[\cA(G) = i] \propto \exp\Big(\dfrac{\eps \cdot q_i(G)}{2\Delta}\Big)$. Algorithm $\cA$ is $\eps$-node-private.
\end{theorem}

\end{document}